\documentclass{article}
\usepackage{epsfig}
\usepackage{fullpage}
\usepackage{dsfont}
\usepackage{amsmath}
\usepackage{amsthm}
\usepackage{color} 
\usepackage[ ruled, lined, noresetcount]{algorithm2e}

\newtheorem{lemma}{Lemma}

\newtheorem{theorem}[lemma]{Theorem}

\newtheorem{claim}[lemma]{Claim}

\def\comic#1#2#3{\parbox{#1}{\centering\includegraphics[width=#1]{#2}\\{\footnotesize #3}}}
\def\comicII#1#2{\parbox{#1}{\centering\includegraphics[width=#1]{#2}}}

\newcommand{\old}[1]{{}}

\newcommand{\degree}{\ensuremath{^\circ}}

\title{Minimum Covering with Travel Cost\thanks{A preliminary
extended abstract of this paper appears in \cite{fms-mctc-09}. A
4-page abstract based on Sections 3 and 4 of this paper appeared in
the informal and non-selective workshop ``EuroCG'', March 2009
\cite{fs-lctnwdv-09}.}}

\author{S{\'a}ndor P.~Fekete
      \thanks{Department of Computer Science, TU Braunschweig,
Germany. Email {\tt \{s.fekete,c.schmidt\}@tu-bs.de}.}
    \and
  Joseph S.~B.~Mitchell\thanks{Department of Applied Mathematics and Statistics, Stony Brook University, USA. Email {\tt jsbm@ams.stonybrook.edu}.} \thanks{Partially supported by
the National Science Foundation (CCF-0528209, CCF-0729019), Metron
Aviation, and NASA Ames.} 
\and
    Christiane Schmidt \footnotemark[2] \thanks{Partially supported by DFG Focus Program ``Algorithm Engineering'' (SPP 1307)
    project ``RoboRithmics'' (Fe 407/14-1). }}

\begin{document}
\date{}

\maketitle

\markboth{}{}


\begin{abstract}

Given a polygon and a visibility range, the Myopic Watchman Problem
with Discrete Vision (MWPDV) asks for a closed path $P$ and a set of
scan points ${\cal S}$, such that (i) every point of the polygon  is
within visibility range of a scan point; and (ii) path length plus
weighted sum of scan number along the tour is minimized.
Alternatively, the bicriteria problem (ii') aims at minimizing both
scan number and tour length. We consider both lawn mowing (in which
tour and scan points may leave $P$) and milling (in which tour, scan
points and visibility must stay within $P$) variants for the MWPDV;
even for simple special cases, these problems are NP-hard.

We show that this problem is NP-hard, even for the special cases of
rectilinear polygons and $L_\infty$ scan range 1, and negligible
small travel cost or negligible travel cost.  For rectilinear MWPDV
milling in grid polygons we present a 2.5-approximation with unit
scan range; this holds for the bicriteria version, thus for any
linear combination of travel cost and scan cost. For grid polygons
and circular unit scan range, we describe a bicriteria
4-approximation. These results serve as stepping stones for the
general case of circular scans with scan radius $r$ and arbitrary
polygons of feature size $a$,
%
%
for which we extend the underlying ideas to a
$\pi(\frac{r}{a}+\frac{r+1}{2})$ bicriteria approximation algorithm.
Finally, we describe approximation schemes for MWPDV lawn mowing and
milling of grid polygons, for fixed ratio between scan cost and
travel cost.


\end{abstract}

\section{Introduction}
Covering a given polygonal region by a small set of disks or squares
is a problem with many applications. Another classical problem is
finding a short tour that visits a number of objects. Both of these
aspects have been studied separately, with generalizations motivated
by natural constraints.

In this paper, we study the combination of these problems,
originally motivated by challenges from robotics, where accurate
scanning requires a certain amount of time for each scan; obviously,
this is also the case for other surveillance tasks that combine
changes of venue with stationary scanning. The crucial constraints
are (a) a limited visibility range, and (b) the requirement to stop
when scanning the environment, i.e., with vision only at discrete
points. These constraints give rise to the {\em Myopic Watchman
Problem with Discrete Vision} (MWPDV), the subject of this paper.

For a scan range that is not much bigger than the feature size of
the polygon, the MWPDV combines two geometric problems that allow
approximation schemes (minimum cover and TSP). This makes it
tempting to assume that combining two approximation schemes will
yield a polynomial-time approximation scheme (PTAS), e.g., by using
a PTAS for minimum cover (Hochbaum and Maass
\cite{hm-ascppipvlsi-85}), then a PTAS for computing a tour on this
solution. As can be seen from Figure~\ref{ptasx2}(a) and (b), this
is not the case; moreover, an optimal solution depends on the
relative weights of tour length and scan cost. This turns the task
into a bicriteria problem; the example shows that there is no
simultaneous PTAS for both aspects. As we will see in
Sections~\ref{mwpdv_rect} and \ref{AppxCircular}, a different
approach allows a simultaneous constant-factor approximation for
both scan number and tour length, and thus of the combined cost. We
show in Section~\ref{ptas}, a more involved integrated guillotine
approach allows a PTAS for combined cost in the case of a fixed
ratio between scan cost and travel cost.

\begin{figure}[h]
\centering
\begin{minipage}{.45\textwidth}
    \comic{\textwidth}{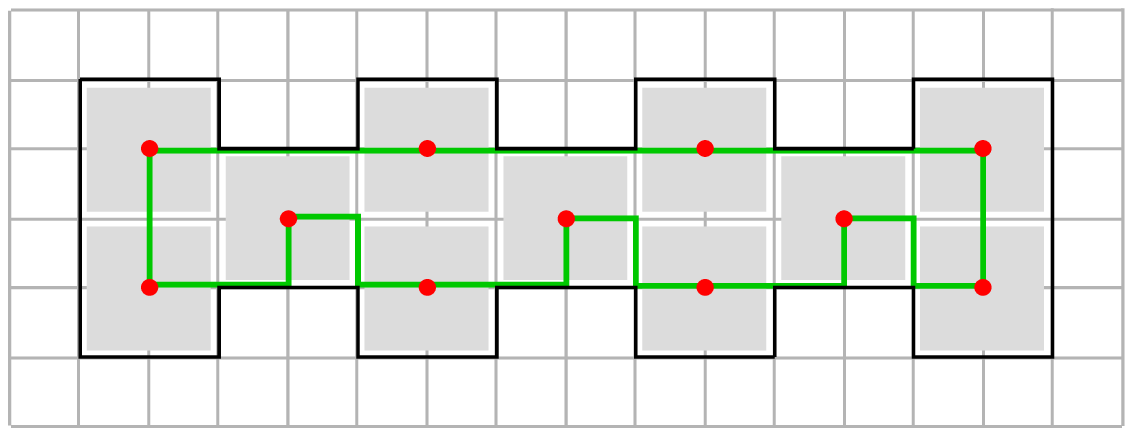}{(a)}
    \\
    \comic{\textwidth}{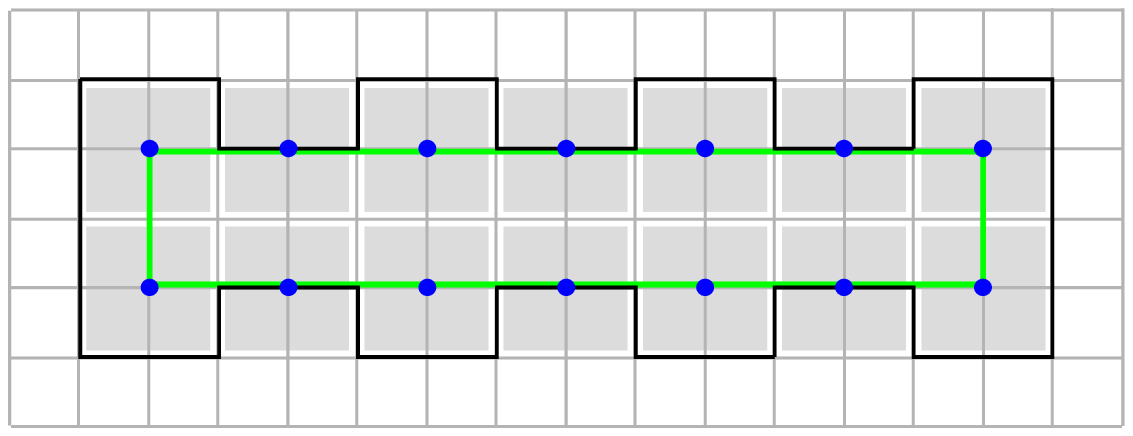}{(b)}
    \end{minipage}
    \hfill
     \comic{.45\textwidth}{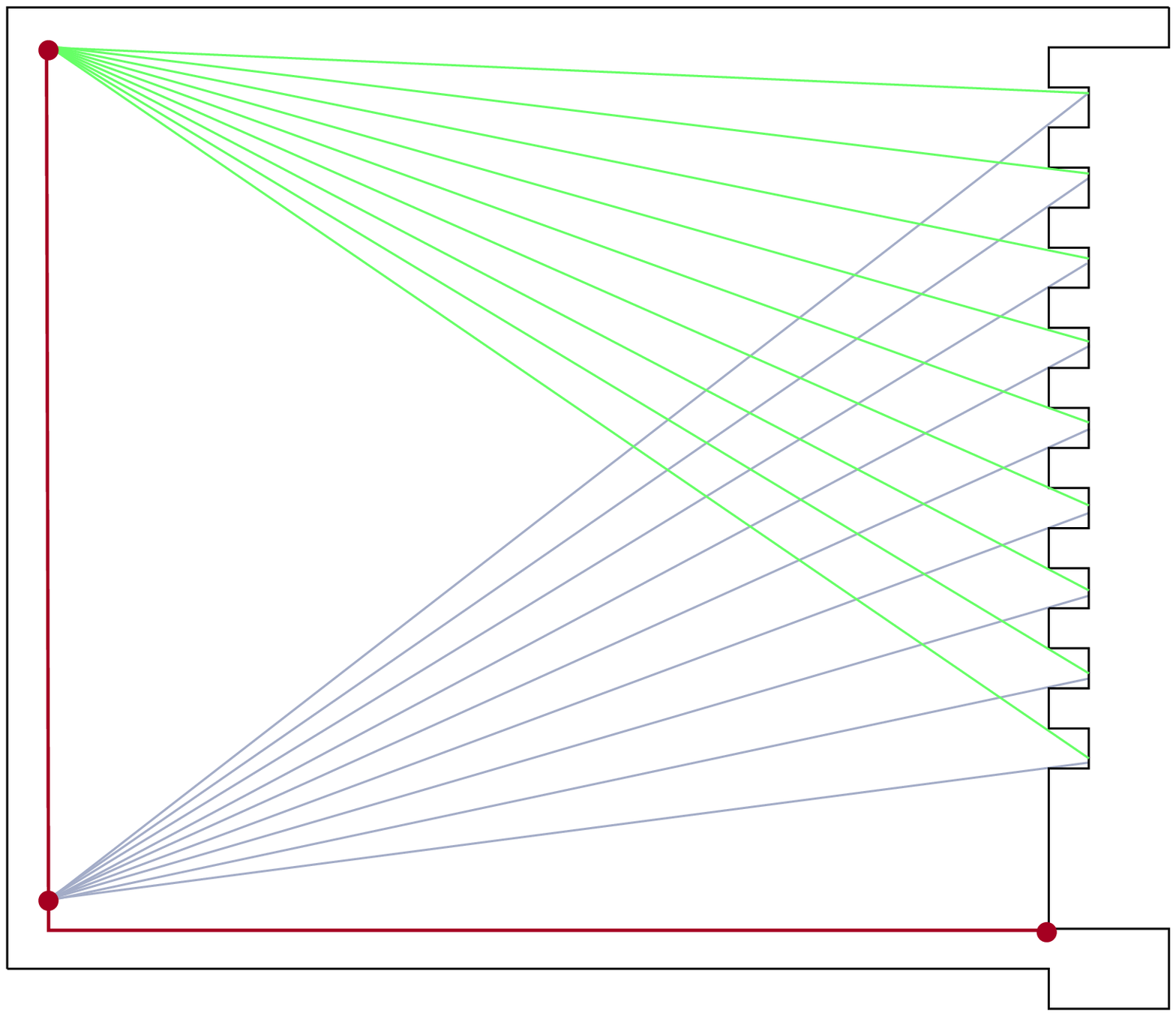}{(c)}
     \caption{\label{ptasx2}\small (a) An MWPDV solution
with a minimum number of scans; (b) an MWPDV solution with a minimum
tour length.
 (c) A minimum guard cover may involve
scan points that are not from an obvious set of candidate points.}
\end{figure}

A different kind of difficulty is highlighted in
Figure~\ref{ptasx2}(c), where an example called ``wells''
(\cite{amp-lgvcp-10}) is illustrated: For a visibility range $r$
that is large compared to the size of the niches (the ``wells''), it
may be quite hard to determine a guard cover of small size, since it
is not clear what should be a small set of candidate guard locations
that suffice for coverage. Each ``well'' is covered by a pair of
points, whose locations need not correspond, e.g., to vertices in
the arrangement of visibility polygons of the vertices of the
polygon. In fact, because of this difficulty, there is no known
nontrivial (with factor $o(n)$) approximation for minimum guard
cover in simple $n$-gons.  With an assumption about a sufficient set
of guard points (e.g., vertices of the polygon or a grid within the
polygon), an $O(\log OPT)$-approximation by Efrat and
Har-Peled~\cite{eh-ggt-06} is known. Also note that the optimal
solution for our problem may change significantly with the relative
weights between tour length and scan cost: If the tour length
dominates the number of scans in the objective function, an optimal
tour can be forced to follow the row of niches on the right in the
figure. We will show in Section 6 how to obtain a constant-factor
approximation for a bounded value $\frac{r}{a}$, $a$ being the
minimum side length of $P$.

 {\bf Related Work.} Closely related to practical problems of
searching with an autonomous robot is the classical theoretical
problem of finding a {\em shortest watchman tour}; e.g., see
\cite{cn-owr-88,cn-swrsp-91}. Planning an optimal set of scan points
(with unlimited visibility) is the {\em art gallery problem}
\cite{o-agta-87}. Finally, visiting all grid points of a given set
is a special case of the classical {\em Traveling Salesman Problem}
(TSP); see \cite{ips-hpgg-82}. Two generalizations of the TSP are
the so-called {\em lawn mowing} and {\em milling problems}: Given a
cutter of a certain shape, e.g., an axis-aligned square, the {\it
milling problem} asks for a shortest tour along which the (center of
the) cutter moves, such that the entire region is covered and the
cutter stays inside the region at all times. Clearly, this takes
care of the constraint of limited visibility, but it fails to
account for discrete visibility. At this point, the best known
approximation method for milling is a 2.5-approximation
\cite{afm-aalmm-00}. Related results for the TSP with neighborhoods
(TSPN) include \cite{dm-aatnp-03,m-ptast-07}; further variations
arise from considering online scenarios, either with limited vision
\cite{bgs-euped-01} or with discrete vision
\cite{fs-pedv-09,fs-pedv-08}, but not both. The discrete visibility
is intrinsic to the art gallery problem, but no tour is considered
here. For this problem neither constant-factor approximation
algorithms nor exact solution methods are known, recent results
include an algorithm based on linear programming that provides lower
bounds on the necessary number of guards in every step and---in case
of convergence and integrality---ends with an optimal solution by
Baumgartner et al.~\cite{bfks-esbgagp-10}. Finally,
\cite{aab-mccps-06} consider covering a set of points by a number of
scans, and touring all scan points, with the objective function
being a linear combination of scan cost and travel cost; however,
the set to be scanned is discrete, and scan cost is a function of
the scan radius, which may be small or large.

For an online watchman problem with unrestricted but discrete
vision, Fekete and Schmidt \cite{fs-pedv-09} present a comprehensive
study of the milling problem, including a strategy with constant
competitive ratio for polygons of bounded feature size
%
%
and with the assumption that each edge of the polygon is fully
visible from some scan point. For limited visibility range, Wagner
et al.~\cite{blw-mvpdr-00} discuss an online strategy that chooses
an arbitrarily uncovered point on the boundary of the visibility
circle and backtracks if no such point exists. For the cost they
only consider the length of the path used between the scan points,
scanning causes no cost. Then, they can give an upper bound on the
cost as a ratio of total area to cover and squared radius.

{\bf Our Results.} On the positive side, we give a 2.5-approximation
for the case of grid polygons and a rectangular range of unit-range
visibility, generalizing the 2.5-approximation by Arkin, Fekete, and
Mitchell \cite{afm-aalmm-00} for continuous milling. The underlying
ideas form the basis for more general results: For circular scans of
radius $r=1$ and grid polygons we give a 4-approximation. Moreover,
for circular scans of radius $r$ and arbitrary polygons of feature
size $a$, we extend the underlying ideas to a
$\pi(\frac{r}{a}+\frac{r+1}{2})$-approximation algorithm. All these
results also hold for the bicriteria versions, for which both scan
cost and travel cost have to be approximated simultaneously.
Finally, we present a PTAS for MWPDV lawn mowing and a PTAS for
MWPDV milling, both for the case of fixed ratio between scan cost
and travel cost.

The rest of the paper is organized as follows. In the following
Section~\ref{notation} we give the notation and formally define the
Myopic Watchman Problem with Discrete Vision. Section~\ref{hardness}
provides a NP-hardness proof. Approximation algorithms for grid
polygons and rectangular unit scan range, grid polygons and circular
unit scan range as well as for general polygons and circular scan
range are presented in Sections~\ref{mwpdv_rect}, \ref{AppxCircular}
and \ref{gencirc}, respectively. A description of polynomial-time
approximation schemes for both the lawn mowing and the milling
variant are given in Section~\ref{ptas}. In the final
Section~\ref{conclusion} we discuss possible implications and
extensions.

\section{Notation and Preliminaries}\label{notation}
We are given a polygon $P$. In general, $P$ may be a polygon with
holes; in Sections 3, 4 and 5, $P$ is an axis-parallel polygon with
integer coordinates.

Our robot, $R$, has discrete vision, i.e., it can perceive its
environment when it stops at a point and performs a scan, which
takes $c$ time units. 
      From a scan point $p$, only a ball of radius $r$ is
visible to $R$, either in $L_{\infty}$- or $L_2$-metric. A set
${\cal S}$ of scan points {\em covers} the polygon $P$, if and only
if for each point $q\in P$ there exists a scan point $p\in {\cal S}$
such that $q$ sees $p$ (i.e., $qp\subset P$) and $|qp|\leq r$.

We then define the {\em Myopic Watchman Problem with Discrete
Vision} (MWPDV) as follows: Our goal is to find a tour $T$ and a set
of scan points ${\cal S}(T)$ that covers $P$, such that the total
travel and scan time is optimal, i.e., we minimize $t(T) = c\cdot
|{\cal S}(T)| + L(T)$, where $L(T)$ is the length of tour~$T$.
Alternatively, we may consider the bicriteria problem, and aim for a
simultaneous approximation of both scan number and tour length.

\section{NP-Hardness}\label{hardness}
Even the simplest and extreme variants of MWPDV lawn mowing are
still generalizations of NP-hard problems.

\begin{theorem}\label{np}
$\;$\\\vspace*{-0.5cm}
\begin{enumerate}
\item[(1)] The MWPDV is NP-hard, even for polyominoes and small or no
scan cost, i.e., $c\ll1$ or $c=0$.

\item[(2)] The MWPDV is NP-hard, even for polyominoes and small travel
cost, i.e., $c\gg1$.

\item[(3)] The MWPDV is NP-hard, even for polyominoes and no travel
cost, i.e., $t(T)=|{\cal S}|$.
\end{enumerate}
\end{theorem}

\begin{proof}
The first claim is a result of the hardness of minimum cost milling,
see \cite{afm-aalmm-00}. The second claim is an easy consequence of
the NP-hardness of {\tt Hamiltonicity of Grid Graphs} (HGG)
\cite{ips-hpgg-82}: Given an instance $G$ of HGG with $n$ vertices,
turn it into an instance of MWPDV by scaling the grid graph $G$ by a
factor of two, and replacing each grid point of $G$ by a 2x2-square.
This yields a canonical set of $n$ scan points that is contained in
any optimal WMPDV tour; traveling these with a tour length $2n$ is
possible if and only the graph $G$ is Hamiltonian.

The third claim is closely related to a minimum cover problem by
visibility discs; however, the MWPDV requires that the scan points
must be inside of the polygonal region. We give a proof along the
lines of \cite{bf-agdp}, based on a reduction of the NP-hard problem
{\tt Planar 3SAT}, a special case of 3SAT in which the
variable-clause incidence graph $H$ is planar. As a first step, we
construct an appropriate planar layout of the graph $H$, e.g., by
using the method of Rosenstiehl and Tarjan~\cite{rt-rplbo-86}. This
layout is turned into a grid polygon by representing the variables,
the clauses and the edges of $G$. An example for the variable
component is given in Figure~\ref{variable}.

\begin{figure}
\centering
    \hspace*{0.07\textwidth}
    \comicII{.4\textwidth}{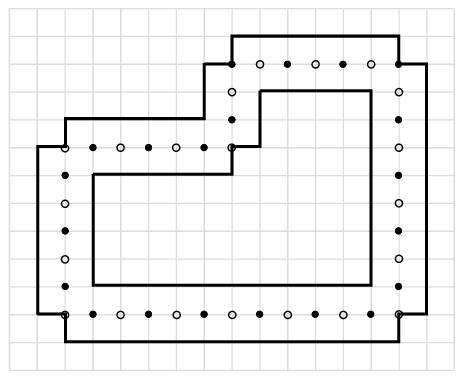}
    \hspace*{0.06\textwidth}
    \comicII{.4\textwidth}{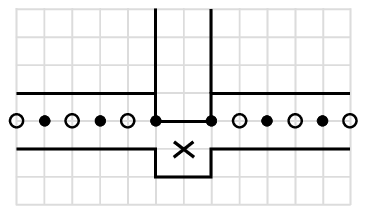}
    \hspace*{0.07\textwidth}
    \\
    \hspace*{0.07\textwidth}
    \comicII{.4\textwidth}{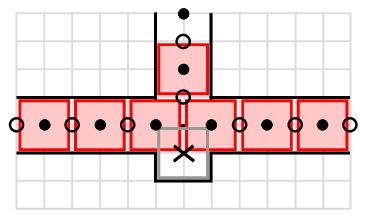}
    \hspace*{0.06\textwidth}
    \comicII{.4\textwidth}{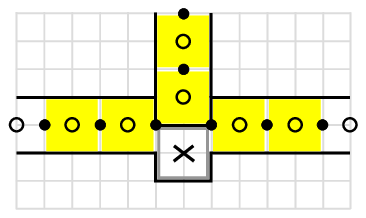}
    \hspace*{0.07\textwidth}

 \caption{\label{variable}\small Top row: The polygonal piece for a variable and the connection of an edge corridor. Bottom row: a placement corresponding to
``true'' (left), and a placement corresponding to ``false'' (right).
The light gray scan is used in both cases. }
\end{figure}

The variable gadgets allow two ways for locating a minimum number of
scan points. The first (the black points in Figure \ref{variable})
relates to a setting of ``true'', the other (the circles in Figure
\ref{variable}) to a setting of ``false''. For the variable setting
``true'' the scan squares are pushed further into the edge corridor.
These edge corridors are similar to the variable-circles---bendings
are done accordingly. In order to have the same number of points and
circles, edge corridors may be added, that do not end in another
polygonal piece, but assure this parity (with circles at the edge
corridor).

A clause component is given in Figure~\ref{clause}. Edge corridors
of the three associated variables (each with the appropriate truth
setting) meet in the polygonal piece for the clause (dark gray in
Figure~\ref{clause}). If and only if the clause is satisfied, i.e.,
if at least in one edge corridor a scan square is placed at a black
point, three additional scans suffice to cover this polygonal piece.
Otherwise, four scans are necessary.

Given the components defined above we can compute the parameter $k$,
the number of scan squares necessary to cover the entire resulting
polygon $P$. $k$ is polynomial in the number of vertices of $G$ and
part of the input. All vertices of the resulting $P$ have integer
coordinates of small size, their number is polynomial in the number
of vertices of $G$. This shows that the problem is NP-hard. \qed

\begin{figure}[h]
\centering
    \comic{.3\textwidth}{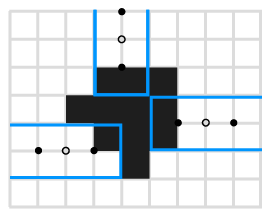}{(a)}
    \hfill
    \comic{.3\textwidth}{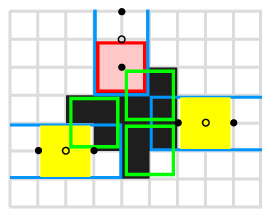}{(b)}
    \hfill
    \comic{.3\textwidth}{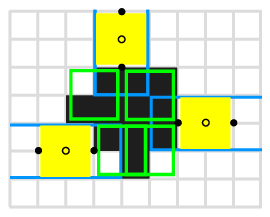}{(c)}
 \caption{\label{clause}\small A clause component with the polygonal piece (dark gray) and the three according edge corridors (blue)(a).
 A placement corresponding to ``true'' (b), and a placement corresponding to
``false'' (c). The green squares are the scans necessary inside the
clause. }
\end{figure}

\end{proof}

\section{Approximating Rectilinear MWPDV Milling for Rectangular Visibility
Range}\label{mwpdv_rect} As a first step (and a warmup for more
general cases), we give an approximation algorithm for rectilinear
visibility range in rectilinear grid polygons.

The following lemma allows us to focus on visiting and scanning at
grid points.

\begin{lemma}\label{gridscans}
For a rectilinear pixel polygon $P$ there exists an optimum myopic
watchman tour $T^*$ such that all scan points are located on grid
points:
\begin{equation}
\exists T^* \forall (x_s, y_s) \in {\cal S}(T^*): x_s \in \mathds{Z}
\wedge y_s \in \mathds{Z}.
\end{equation}
\end{lemma}
\begin{proof}
Let $T$ be an optimal tour, with scan points not located on grid
points. Consider the vertical and horizontal strips of pixels in $P$
of maximal length. W.l.o.g., we start with the horizontal strips.
For every strip, we shift the scans, such that the x-coordinates are
integers (starting from the boundary, i.e., with distance $1$ to the
boundary if possible, and away from non-reflex corners). The tour
will not be longer, we cover not less; in case we are able to reduce
the number of scans per strip by one we have a contradiction to $t$
being optimal. After applying this to all horizontal strips we
proceed analogously for the vertical strips. Hence, we have an
optimal tour, with all scan points located on grid points.
\end{proof}

Our approximation proceeds in two steps:
\begin{enumerate}
\item[(I)] Construct a set of scan points that is not larger than 2.5 times
a minimum cardinality scan set.
\item[(II)]Construct a tour that contains all constructed scan points and
that does not exceed 2.5 times the cost of an optimum milling tour.
\end{enumerate}

The main idea for the second step is based on the $O(n \log n)$-time
2.5-approximation algorithm for milling from Arkin et
al.~\cite{afm-aalmm-00}. The resulting tour consists of three parts,
see Figure \ref{milling} for an example, $L_{OPT}$ being the optimal
milling tour length:
\begin{itemize}
\item[(1)] a ``boundary'' part: $B \subset P$ is the inward
offset region of all points within P that are feasible placements
for the center of the milling cutter. For a milling problem, $B$ is
connected. Tracing the boundary $\delta B$ of B, let $P_{\delta B}$
denote the region milled by this route. ($\delta B$ may not be
connected (if $P$ features holes), the pieces are $\delta B_i$.) The
length of $\delta B$, $L_{\delta B}$ is a lower bound on $L_{OPT}$.
\item[(2)] a ``strip'' part: $P_{int} := P \backslash P_{\delta B}$---if nonempty---can
be covered by a set of $k$ horizontal strips $\Sigma_i$. The
$y$-coordinates of two strips differ by multiples of $2$. Then, let
$L_{str} = \sum_{i=1}^k L_{\Sigma_{i}}$ and this is again a lower
bound on the length of an optimal milling tour: $L_{OPT} \geq
L_{str}$.
\item[(3)] a ``matching'' part: the strips and the boundary tour have
to be combined for a tour. For that purpose, consider the endpoints
of strips on $\delta B_i$: every $\delta B_i$ contains an even
number of such endpoints. Hence, every $\delta B_i$ is partitioned
into two disjoint portions, $M_1(\delta B_i)$ and $M_2(\delta B_i)$.
Using the shorter of these two ($M_*(\delta B_i)$) for every $\delta
B_i$ we obtain for the combined length, $L_M$: $L_M \leq L_{str}/2
\leq L_{OPT}/2$.
\end{itemize}
The graph with endpoints of strip lines plus the points where strip
line touches a $\delta B_i$ as vertices is connected by three types
of edges---the center lines of the strips, the $\delta B_i$s and the
$M_*(\delta B_i)$s. Every vertex has degree 4, hence, an Eulerian
tour gives a feasible solution.

\begin{figure}[h]
\centering
\includegraphics[width=.6\textwidth]{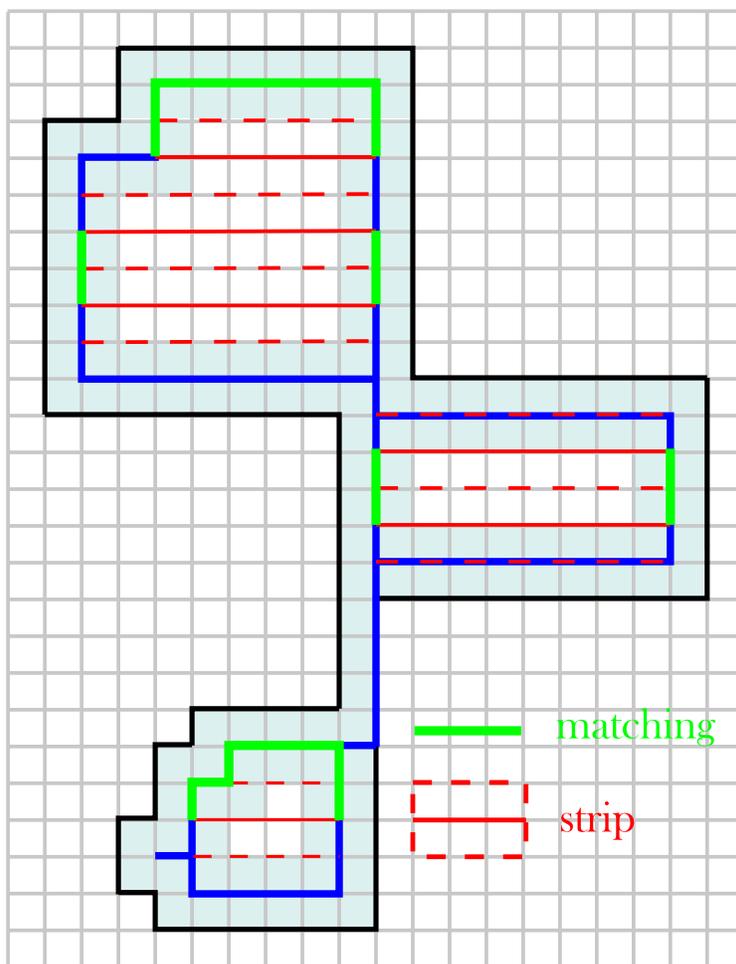}
 \caption{\label{milling}\small A polygon $P$, consisting of pixels.
 The parts for the approximative milling tour are indicated:
 region $P_{\delta B_i}$ is shaded (light blue), $\delta B_i$ is highlighted in blue.}
\end{figure}

Now we describe how to construct a covering set of scan points:
\begin{enumerate}
\item Let $S_{4e}$ be the ``even quadruple'' centers of all 2x2-squares that are
fully contained in $P$, and which have two even coordinates.
\item Remove all 2x2-squares corresponding to $S_{4e}$ from
$P$; in the remaining polyomino $P_{4e}$, greedily pick a maximum
disjoint set $S_{4o}$ of ``odd quadruple'' 2x2-squares.
\item Remove all 2x2-squares corresponding to $S_{4o}$ from
$P_{4e}$; greedily pick a maximum disjoint set $S_{3}$ of ``triple''
2x2-squares that cover 3 pixels each in the remaining polyomino
$P_{4e,4o}$,
\item Remove all 2x2-squares corresponding to $S_3$ from
$P_{4e,4o}$; in the remaining set $P_{4e,4o,3}$ of pixels, no three
can be covered by the same scan. Considering edges between pixels
that can be covered by the same scan, pick a minimum set of
(``double'' $S_2$ and ``single'' $S_1$) scans by computing a maximum
matching.
\end{enumerate}


\begin{figure}[h!]
\centering
\includegraphics[width=\columnwidth]{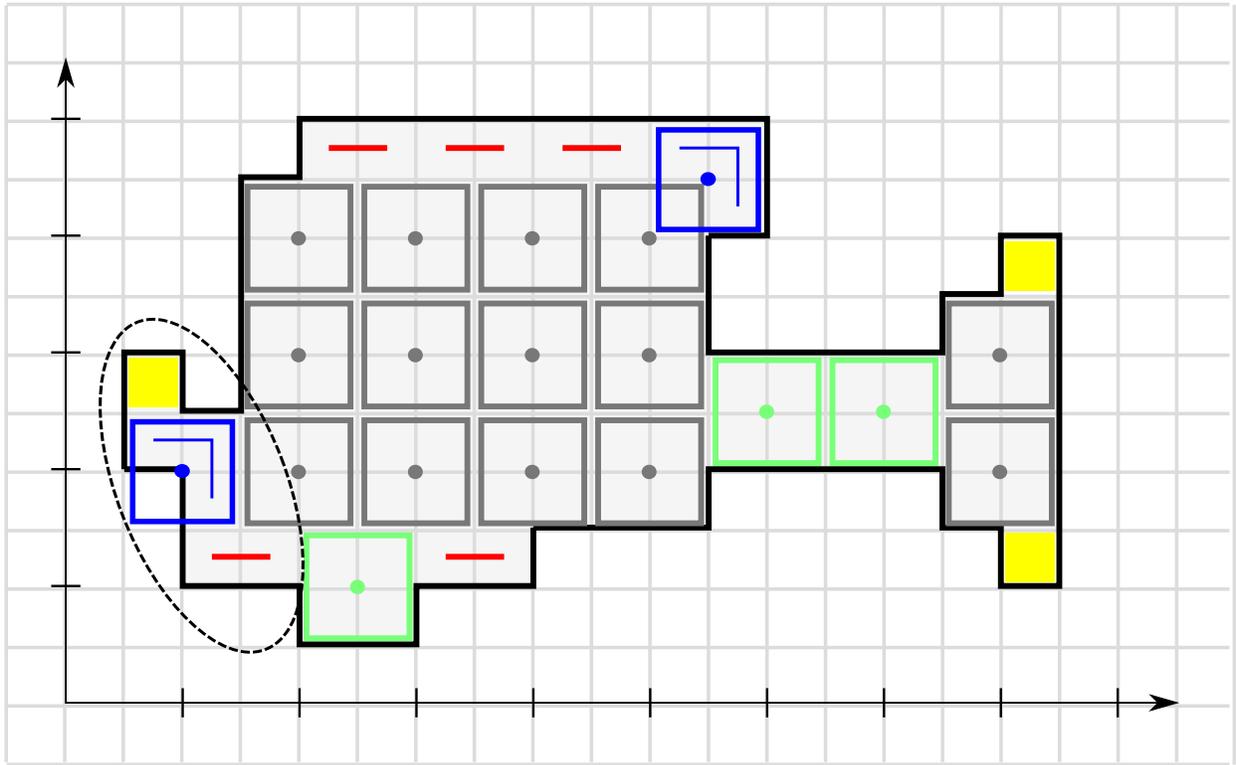}
 \caption{\label{example}\small An example for our approximation method:
The set of ``even quadruple''scans is shown in grey; the ``odd
quadruple'' scans are green. A possible (greedy!) set of ``triple''
scan is shown in blue, leaving the maximum matching (and the
corresponding ``double scans'') shown in red. The leftover single
pixels are yellow. The ellipse indicates a part that is covered by
three scans instead of two: the triple scan with adjacent single and
double scans could be covered by two triple scans.}
\end{figure}

\begin{claim}
The total number of scans is at most 2.5 times the size of a minimum
cardinality scan set.
\end{claim}

\begin{claim}
All scan points lie on a 2.5-approximative milling tour.
\end{claim}

For the first claim, let $s_{\min}$ be the size of a cardinality
scan set. Observe that $S_{4e}\cup S_{4o}$ corresponds to a set of
disjoint 2x2-squares that are fully contained in $P$; hence, it
follows from a simple area argument that $|S_{4e}\cup S_{4o}|\leq
s_{\min}$.

When considering the set of pixels in $P_{4e,4o}$, not more than
three can be covered by the same scan; in $P_{4e,4o,3}$ we compute
an optimal solution. This implies that the only way to get a smaller
cover in  $P_{4e,4o}$ is to change the choice of triple scans; this
means that the pixels covered by a triple scan have to be allocated
differently. Taking into account that $P_{4e,4o}$ does not contain
any 2x2-squares, a simple case analysis shows that only two possible
improvements are possible:
\begin{itemize}
\item replacing a triple, a double and a single by two triples (as
shown in the left part of Figure~\ref{example}), or
\item replacing a triple and two singles by a triple and a double.
\end{itemize}
With $s^4_{\min}$ being the minimum number of scans required for
covering $P_{4e,4o}$, we get $|S_{3}\cup S_{2}\cup S_{1}|\leq
\frac{3}{2} s^4_{\min}\leq \frac{3}{2} s_{\min}$. In total we get a
solution with not more than $\frac{3}{2} s_{\min}$ scans.

For the second claim, we use a milling tour constructed as in
\cite{afm-aalmm-00} and described above: Choosing the strips to be
centered on even $y$-coordinates allows us to visit all scan points
in $S_{4e}$. Clearly, all pixels in $P_{4e}$ are adjacent to the
boundary of other scans involve boundary pixels, allowing them to be
visited along the ``boundary'' part. (One minor technical detail is
shown in Figure~\ref{nrv}: In order to visit the center of a triple
scan, we need to reroute the boundary part of the tour to run
through a reflex vertex; this does not change the tour length.)

\begin{figure}[h]
\centering
    \comicII{.3\textwidth}{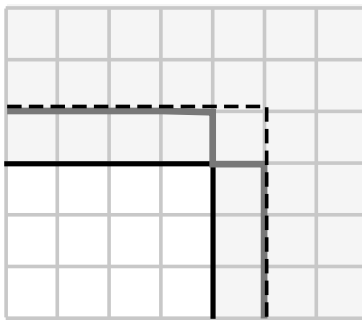}
 \caption{\label{nrv}\small Modifying the boundary part of $W(t)$:
 at a reflex vertex, the gray path is used instead of $\delta B_i$.}
\end{figure}

This concludes the proof.  We summarize:

\begin{theorem}\label{app}
A polyomino $P$ allows a MWPDV with rectangular vision solution that
contains at most 2.5 times the minimum number of scans necessary to
scan the polygon, and has tour length at most 2.5 times the length
of an optimum milling tour.
\end{theorem}

\section{Approximating Rectilinear MWPDV Milling for Circular Visibility
Range}\label{AppxCircular}

In this section we give an approximation algorithm for MWPDV milling
in case of a circular range of visibility. Again, $P$ is a
polyomino, we consider $r=1$ and the tour length is measured
according to the $L_1$-distance. The approximation is for the case
of no given starting point or in case of a given starting point
located on a grid point (in the polyomino or on its boundary).

When considering a circular scan range, one additional difficulty
are boundary effects of discrete scan points: While continuous
vision allows simply sweeping a corridor of width $2r$ by walking
down its center line, additional cleanup is required for the gaps
left by discrete vision; this requires additional mathematical
arguments, see Figure~\ref{circ_diff}.

\begin{figure}[h]
\centering
    \comicII{.7\textwidth}{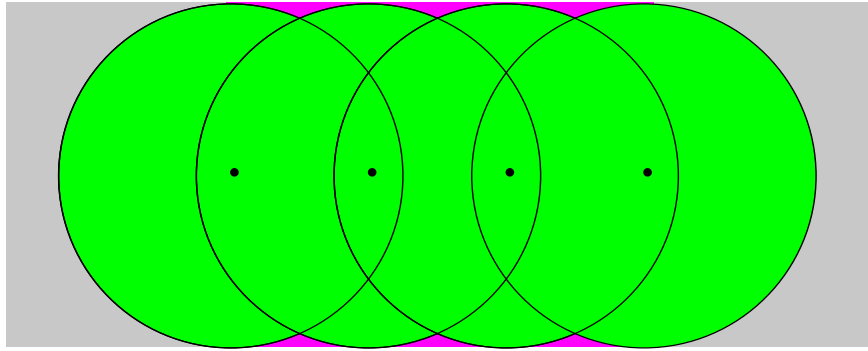}
     \caption{\label{circ_diff}\small A finite number of scan
     points along the center line of a corridor of width $2r$
     does not suffice to cover the corridor with a circular scan range of radius $r$.}
\end{figure}

We overlay the polyomino with a point grid as in
Figure~\ref{interiorscans}, left, i.e., a diagonal point grid with
$L_2$-distance of $\sqrt{2}$ in-between points, and use all points
of this grid that coincide with a grid point in $P$
(Figure~\ref{interiorscans}, left). For the tour $T$ our exploration
strategy starts at a boundary grid point, proceeding
counterclockwise along the boundary and taking a scan at every point
of the overlayed grid located on the boundary. Only using this we
would end up with a tour covering an area around the boundary.

\begin{figure}[h]
\centering
    \comicII{.48\textwidth}{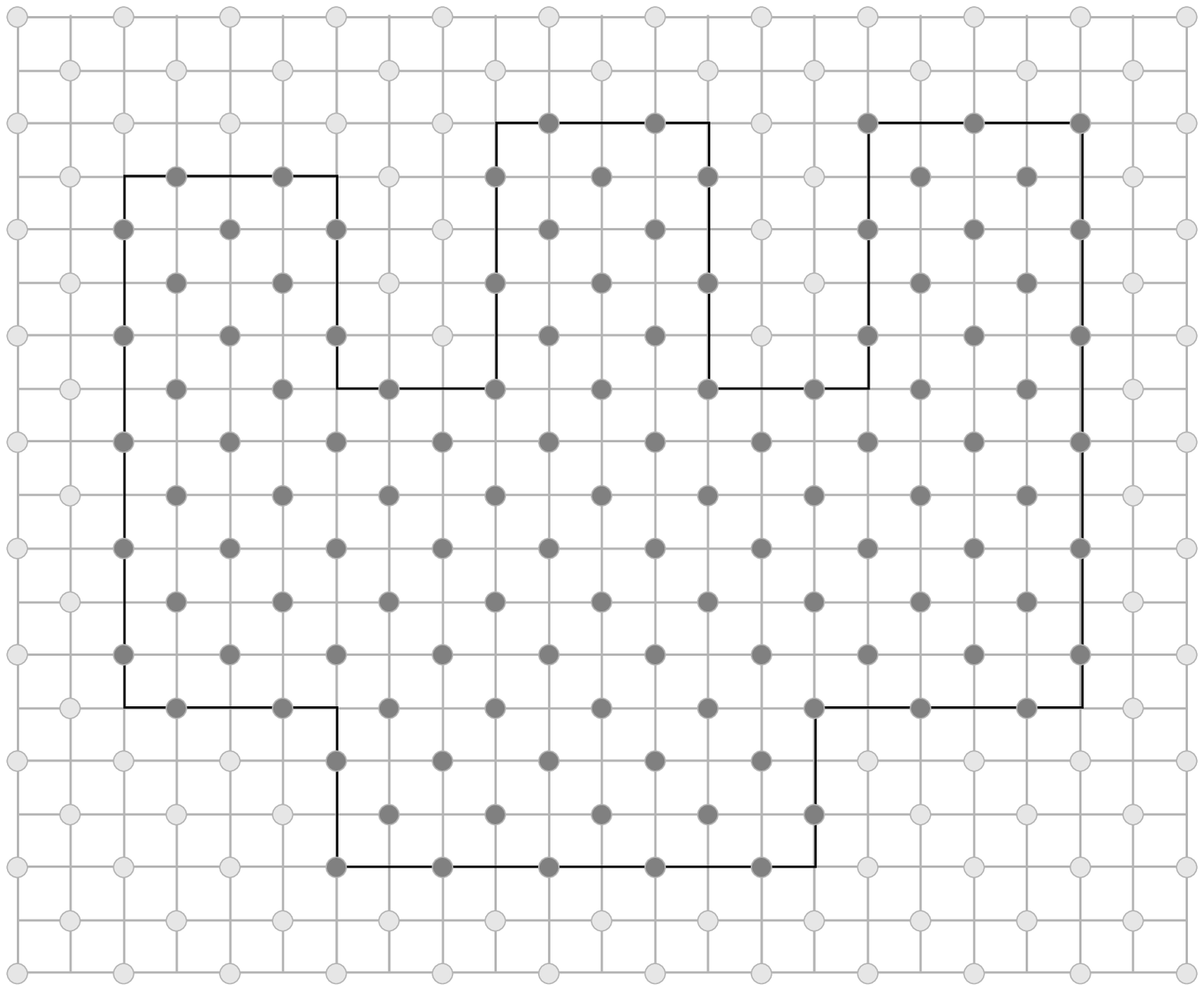}
    \hfill
    \comicII{.48\textwidth}{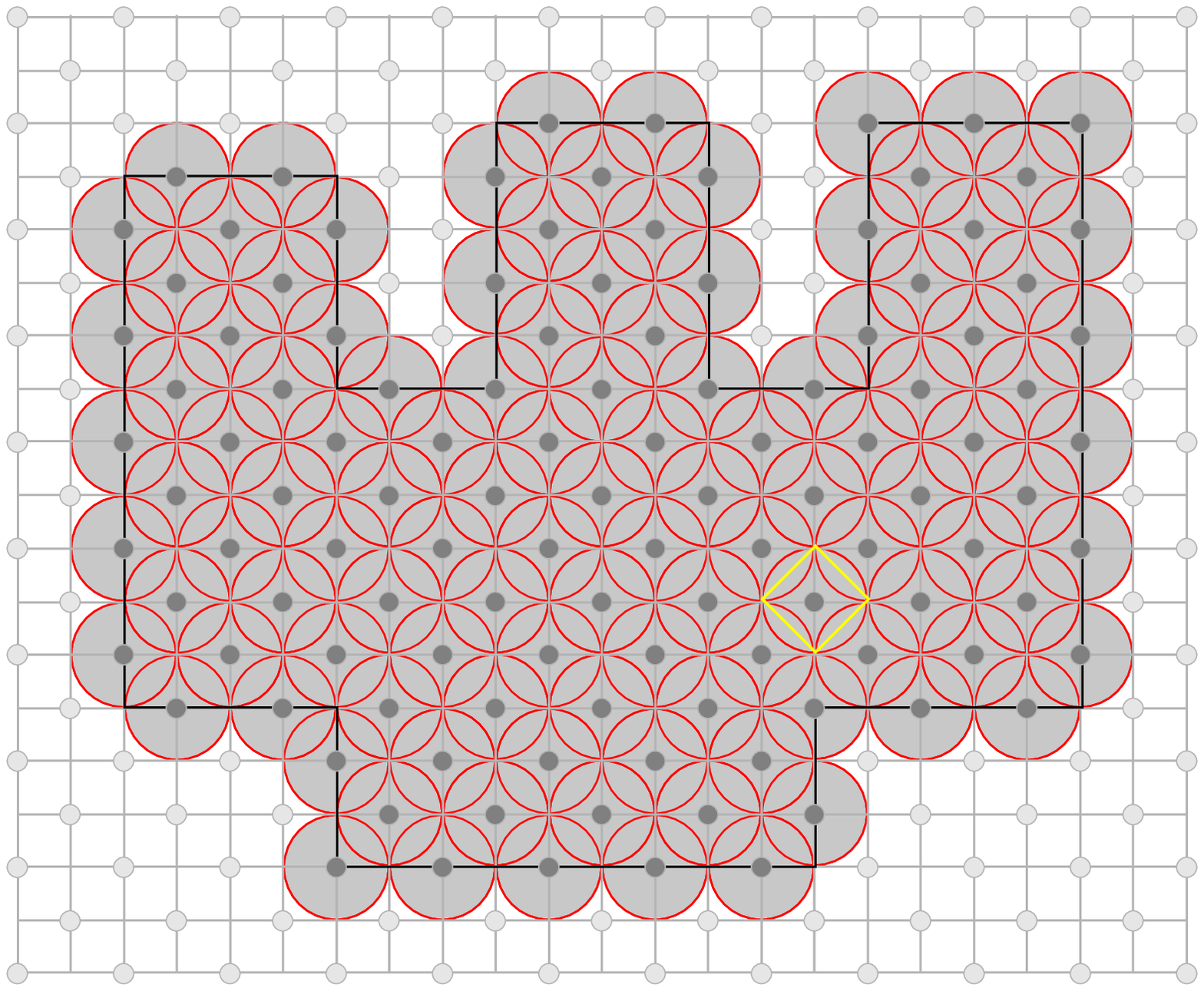}
 \caption{\label{interiorscans}\small Left: Point grid (light gray) with grid points within a polyomino (black) in dark gray.
 Right: Circular visibility ranges of the grid points covering the plane, one square of side length $\sqrt{2}$ is indicated in yellow.}
\end{figure}

For the movement in-between interior scan points and the boundary we
use horizontal strips located on grid lines (and distance $1$ to the
boundary), cp.~Figure~\ref{tour} for an example. In order to combine
these grid lines with the boundary path for a tour, we link strips
to the left boundary. Two of these will in general be linked on the
right-hand side. In case there is an odd number of strips between
the upper and lower boundary of $P$, the scan points located on the
bottommost strip are visited by using a path of ($L_1$-)length $2$,
from the boundary or strips with another y-coordinate of the
leftmost point, cp.~Figure~\ref{tour}, down right. The strips that
get linked are always determined by the leftmost boundary. In case
there are other parts of $P$ that have a left boundary (that is
vertical edges with polygon to the right and the exterior to the
left) whose cardinality of strip lines differ by an odd number, scan
points from the topmost strip are linked by two vertical steps of
length $1$ to the upper boundary, cp.~Figure~\ref{tour}, top right.

\begin{figure}[h]
\centering
    \comicII{.7\textwidth}{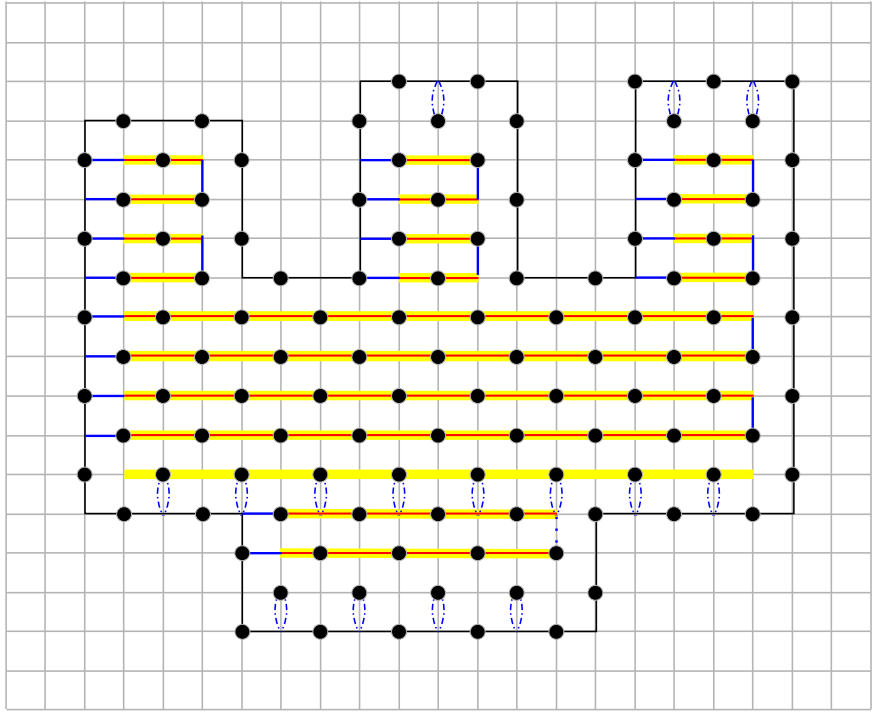}
    \caption{\label{tour}\small A polyomino $P$ with the tour given by our strategy. Scan points are displayed in black.
    The horizontal strips of total length $L_{strips}$ are indicated in yellow,
    the tour is in blue for the links to the boundary and in-between strips
    (solid) as well as for connections of points (dash-dotted),
    parts located on the strips are indicated in red.
    (In-between those parts the tour runs on the boundary.)}
\end{figure}

Hence, we yield a closed tour (linking always two strips we always
end up at the left boundary, the rest is a tour along the boundary
with small loops of length $2$). We still need to show that $P$ is
covered and have to consider the competitive ratio of our strategy.
Let $T$ be the tour determined by our strategy, and T* be an optimal
tour.

\begin{claim}
The scan points the strategy positions cover $P$.
\end{claim}

\begin{proof}
Including the scan points on the boundary the entire polyomino is
covered by the scan point grid defined above: The scan range being a
$L_2$-circle, these scans cover the interior of the polyomino---each
circle covers at least a square of side length $\sqrt{2}$ (with the
vertices located on grid points that do not belong to the overlayed
grid), cp.~Figure~\ref{interiorscans}, right. \qed
\end{proof}

\begin{lemma}\label{T}
$L(T) \leq 4\cdot L(T^*)+8$
\end{lemma}
\begin{proof}
Our tour consists of the tour along the boundary of $P$, the strips
and connections. Let $L_{bound}$ be the length of the boundary,
i.e., $P$'s perimeter. Moreover, we define $\delta_1$:
\begin{itemize}
\item to be $\delta_{Bi}$ (cp.~Section~\ref{mwpdv_rect}) whenever the
corresponding cutter fits into $P$,
\item in-between these $\delta_{Bi}$ a shortest path on the grid,
\item plus (for corridors of $P$ of width $1$) shortest path from
the $\delta_{Bi}$s such that every pixel of $P$ is visited by a part
of $\delta_1$ (a pixel is visited if one of its vertices or edges is
reached).
\end{itemize}
Finally, let the length of all strips be $L_{strips}$.

We need to show that these three elements cover the length of our
chosen path. Movements along the boundary are covered by
$L_{bound}$, movements along strips by $L_{strips}$. The right-hand
side links of two adjacent strips lie on $\delta_1$, the left-hand
side connections can be turned by $90\degree$ (only for charging)
and lie on $\delta_1$ as well. The scan point connecting paths of
length $2$ can be swung open and lie on the strip of the relative
(not used) strip or one part (length $1$) on $\delta_1$ for interior
strip points or endpoints of strips, respectively. Thus, the
combined lengths of these are an upper bound for our tour length.

Using the fact that for an orthogonal polygon of $n$ vertices, $r$
of which are reflex, $n=2r+4$ holds, see \cite{o-agta-87}, Lemma
2.12, we have $L_{bound} = L_{\delta 1} + 8$. Furthermore, with
$L_{str}$ as defined as in Section~\ref{mwpdv_rect}, we have
$L_{strips} \leq 2\cdot L_{str} \leq 2\cdot L(T^*)$. Finally, the
milling argument shows $L_{\delta 1} \leq L(T^*)$. Altogether, we
have: $L(T) \leq L_{bound} + L_{strips} + L_{\delta 1} \leq 2\cdot
L_{\delta 1} + 8 + L_{strips} \leq 4\cdot L(T^*) + 8$. \qed
\end{proof}

\begin{lemma}\label{S}
$|{\cal S}(T)| \leq 4\cdot |{\cal S}(T^*)|$
\end{lemma}
\begin{proof}
Let $N(P)$ be the number of pixels of a polyomino $P$. Moreover, let
$V(N(P))$ be the maximum ratio $\frac{|{\cal S}(T)|}{|{\cal
S}(T^*)|}$ for all polyominos with $N(P)$ pixels. Kershner
\cite{tnccs-k-39}(cp.~\cite{dkdk-t-49}) showed (with the notation
introduced above) that $D(r) = \frac{\pi r^2 |{\cal S}(T^*)|}{N(P)}
\geq \frac{2\sqrt{3} \pi}{9}$, that is, for $r=1$: $|{\cal S}(T^*)|
\geq \frac{2\sqrt{3}}{9} \cdot N(P)$. The placement of scan points
on the diagonal grid allows us to bound the number of scan points by
$2 + (N(P) - 1) = N(P) + 1$. If we define $F(N(P)) :=
\frac{N(P)+1}{\frac{2\sqrt{3}}{9} \cdot N(P)}$, this gives us an
upper bound on $V(N(P))$. $F(n)$ is monotonically decreasing in $n$.
Furthermore, $F(n) = 4 \Leftrightarrow n =
\frac{3\sqrt{3}}{8-3\sqrt{3}} \approx 1.85322$, that is:
\begin{equation}
V(N(P)) \leq F(N(P)) \leq 4  \mbox{ for } N(P) \geq 2
\end{equation}
(For $N(P)=1$ the optimum needs at least one scan, we need at most
$2$.) \qed
\end{proof}

\begin{theorem}
A polyomino $P$ allows a MWPDV solution for a circular visibility
range with $r=1$ that is $4$-competitive.
\end{theorem}
\begin{proof}With Lemmas \ref{T} and \ref{S} we have:
\begin{eqnarray*}
t(T) & = & c\cdot |{\cal S}(T)| + L(T)\\
 & \leq & c\cdot 4\cdot |{\cal S}(T^*)| + 4\cdot L(T^*) + 8\\
 & = & 4\cdot (c\cdot |{\cal S}(T^*)| + L(T^*)) + 8.
\end{eqnarray*} \qed
\end{proof}

\section{Approximating General MWPDV Milling for a Circular Visibility
Range}\label{gencirc} In this section we discuss MWPDV milling for a
circular visibility range $r$ in general polygons. As discussed in
Section 1, even the problem of minimum guard coverage has no known
constant-factor approximation; therefore, we consider a bounded
ratio $r/a$ between visibility range and feature size, i.e., minimum
side length.

 Just as in the rectilinear case for a rectilinear scan
range, see Section~\ref{mwpdv_rect}, our approximation proceeds in
two steps:
\begin{enumerate}
\item[(I)] Construct a set of scan points that is within a constant factor of
a covering set of minimum cardinality.
\item[(II)]Construct a tour that contains all constructed scan points and
is within a constant factor of the cost of an optimum milling tour.
\end{enumerate}


We start with a description of the second step, which will form the
basis for the placement of scan points. Just as in the rectilinear
case, we consider three parts.
\begin{itemize}
\item[(1)] A ``boundary'' part:
Above we described tracing $\delta B$ the boundary of $B$, causing a
tour length of $L_{\delta B}$. Here, we use two ``boundary tours''
within distance of (at most) $\frac{1}{2}r$ and (at most)
$\frac{3}{2}r$ to the boundary, $TR1$ and $TR2$ of length $L_{TR1}$
and $L_{TR2}$, respectively. Then, we have:
\begin{equation}\label{rand}
L_{TR1}+L_{TR2}=2\cdot L_{\delta B} \leq 2\cdot L(T^*)
\end{equation}
(The length of the three tours differs at the vertices: drawing a
line perpendicular there from $TR2$ to $TR1$ the Intercept Theorem
shows that the distance to the diagonal through the vertices of all
tours on $TR1$ is twice as much as on the boundary tour with
distance $r$ to the boundary.)

%
The two ``boundary'' tours allow us to cover a corridor of width
$2r$ with a bounded number of scans, while (\ref{rand}) enables us
to bound the tour length in terms of the optimal length.

\item[(2)] A ``strip'' part:
For the interior we use strips again: $P_{int} := P \backslash
P_{\delta B}$---if nonempty---can be covered by a set of $k_1$
horizontal strips $\Sigma_i^1$. The $y$-coordinates of two strips
differ by multiples of $2r$. We can consider another set of strips,
$\Sigma_i^2$, shifted by $r$. Then, let $L_{str}^j =
\sum_{i=1}^{k_j} L_{\Sigma_{i}^j}$. Similar to the argument for
$L_\infty$, we have  $L_{str}^1 + L_{str}^2 \leq 2\cdot L(T^*)$.
\item[(3)] A ``matching'' part: In order to combine
the two ``boundary parts'' and the two sets of strips for a tour we
add two more set of sections.
\begin{itemize}
\item The center lines of the strips have a distance of $r$ to the boundary,
thus they do not yet touch $TR1$. Consequently, we add $1/2 r$ to
each center line (on each end). For that purpose, we consider the
matchings as defined above. (Consider the endpoints of strips on
$\delta B_i$: every $\delta B_i$ contains an even number of such
endpoints. Hence, every $\delta B_i$ is partitioned into two
disjoint portions, $M_1(\delta B_i)$ and $M_2(\delta B_i)$. Using
the shorter of these two ($M_*(\delta B_i)$) for every $\delta B_i$
we obtain for the combined length, $L_M$: $L_M \leq L_{str}/2 \leq
L(T^*)/2$.) Because two strips are at least a distance of $r$ apart,
the connection to $TR1$ costs less than $1/2\cdot L_M \leq 1/2\cdot
L_{str}/2 \leq L(T^*)/4$.
\item Moreover, we consider the above matchings defined on $TR1$ and
insert the shorter sections of the disjoint parts, ($M^1_*(\delta
B_i)$), for every $\delta B_i$. The Intercept Theorem in combination
with the analogously defined sections on $TR2$ enables us to give an
upper bound of $L_{M^1} \leq L_{str} \leq L(T^*)$.
\end{itemize}
Starting on some point on $TR1$, tracing the strips, and the inner
``boundary'' $TR2$ at once when passing it yields a tour; the above
inequalities show that $L(T) \leq 21/4 \cdot L(T^*)$.
\end{itemize}


Now we only have to take care of (I), i.e., construct an appropriate
set of scan points. For the ``boundary'' part we place scans with
the center points located on $TR1$ and $TR2$ in distance
$\sqrt{3}\cdot r$ (see~Figure~\ref{distance}) if possible, but at
corners we need to place scans, so the minimum width we are able to
cover with the two scans (on both tours) is $a$. For the ``strip''
part the distance of scans is also $\sqrt{3}\cdot r$ on both strip
sets, exactly the distance enabling us to cover a width of $r$,
see~Figure~\ref{distance}.

\begin{figure}[h]
\centering
    \comicII{.5\textwidth}{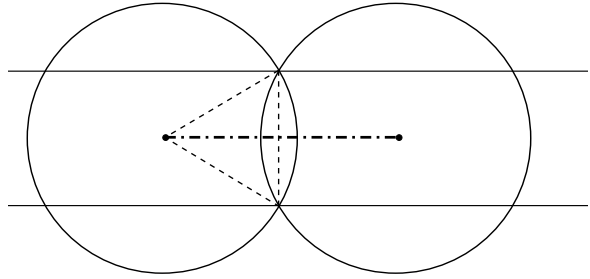}
 \caption{\label{distance}\small The dashed lines indicate a length of $r$,
 so the dash-dotted line has a length of $\sqrt{3} \cdot r$, the distance of two scans.}
\end{figure}

It remains to consider the costs for the scans. We start with the
inner part. Taking scans within a distance of $\sqrt{3}\cdot r$, we
may need the length divided by this value, plus one scan. We only
charge the first part to the strips, the (possible) additional scans
are charged to the ``boundary'' part, as we have no minimum length
of the strips. The optimum cannot cover more than $\pi r^2$ with one
scan. Let $L_{str}=max(L_{str}^1, L_{str}^2)$:
\begin{eqnarray}
|{\cal S}(T^*)| \geq \frac{L_{str}}{\pi r/2},\;\;\; |{\cal S}(T)|
\leq \frac{2L_{str}}{\sqrt{3} \cdot r}\; \Rightarrow \; \frac{|{\cal
S}(T)|}{|{\cal S}(T^*)|} \leq \frac{2L_{str}}{\sqrt{3} \cdot r}
\cdot \frac{\pi r/2}{L_{str}} = \frac{\pi}{\sqrt{3}}
\end{eqnarray}
Finally, we consider the ``boundary''. We assume $L_{\delta B} \geq
1$. So $|{\cal S}(T^*)| \geq \frac{L_{\delta B}}{\pi r/2}$. We may
need to scan within a distance of $a$---on two strips---, need
additional scans and have to charge the scans from the ``strip''
part, hence, this yields: $|{\cal S}(T)| \leq \frac{L_{\delta
B}}{a/2}+1+\frac{L_{\delta B}}{r}$. Consequently, for $r \geq a$:
\begin{eqnarray} \frac{|{\cal S}(T)|}{|{\cal S}(T^*)|} \leq \frac{\pi r}{a} +
\frac{\pi r}{2} + \frac{\pi}{2}\end{eqnarray} 

\begin{theorem}\label{app_l2}
A polygon $P$ allows a MWPDV solution that contains at most a cost
of $max(\frac{21}{4}, \frac{\pi r}{a} + \frac{\pi r}{2} +
\frac{\pi}{2})$ times the cost of an optimum MWPDV solution (for $r
\geq a$).
\end{theorem}

Note that Theorem \ref{app_l2} covers the case from Section
\ref{AppxCircular}; however, instead of the custom-built factor of
$4$ it yields a factor of $2\cdot \pi$.

\section{A PTAS for MWPDV Lawn Mowing}\label{ptas}

We describe in detail here the following special case, and then
discuss how the method generalizes.  Consider a polyomino $P$ (the
``grass'') that is to be ``mowed'' by a $k\times k$ square, $M$.  At
certain discrete set ${\cal S}(T)$ of positions of $M$ along a tour
$T$, the mower is activated (a ``scan'' is taken), causing all of
the grass of $P$ that lies below $M$ at such a position to be mowed.
For complete coverage, we require that $P$ be contained in the union
of $k\times k$ squares centered at points ${\cal S}(T)$.  Between
scan positions, the mower moves along the tour $T$. 

In this ``lawn mower'' variant of the problem, the mower is not
required to be fully inside $P$; the mower may extend outside $P$
and move through the exterior of $P$, e.g., in order to reach
different connected components of $P$. Since $P$ may consist of
singleton pixels, substantially separated, the problem is NP-hard
even for $k=1$, from TSP.

Here we describe a PTAS for the problem.  We apply the
$m$-guillotine method, with special care to handle the fact that we
must have full coverage of $P$.  Since the problem is closely
related to the TSPN~\cite{dm-aatnp-03,m-ptast-07}, we must address
some of the similar difficulties in applying PTAS methods for the
TSP: in particular, a mower centered on one side of a cut may be
responsible to cover portions of $P$ on the opposite side of the
cut.

At the core of the method is a structure theorem, which shows that
we can transform an arbitrary tour $T$, together with a set ${\cal
S}(T)$ of scan points, into a tour and scan-point set, $(T_G,{\cal
S}(T_G))$, that are $m$-guillotine in the following sense: the
bounding box of the set of $k\times k$ squares centered at ${\cal
S}(T)$ can be recursively partitioned into a rectangular subdivision
by ``$m$-perfect cuts''. An axis-parallel cut line $\ell$ is
$m$-perfect if its intersection with the tour has at most $m$
connected components and its intersection with the union of $k\times
k$ disks centered at scan points consists of $m$ disks or ``chains
of disks'' (meaning a set of disks whose centers lie equally spaced,
at distance $k$, along a vertical/horizontal line); see
Figure~\ref{mperfect}.  (The definition of $m$-perfect in
\cite{m-gsaps-99} has a slightly different specification, in terms
of the number of endpoints of the connected components of $T\cap
\ell$, but it is within a constant factor equivalent to what we
define here.)

\begin{figure}[h]
\centering
    \comicII{.7\textwidth}{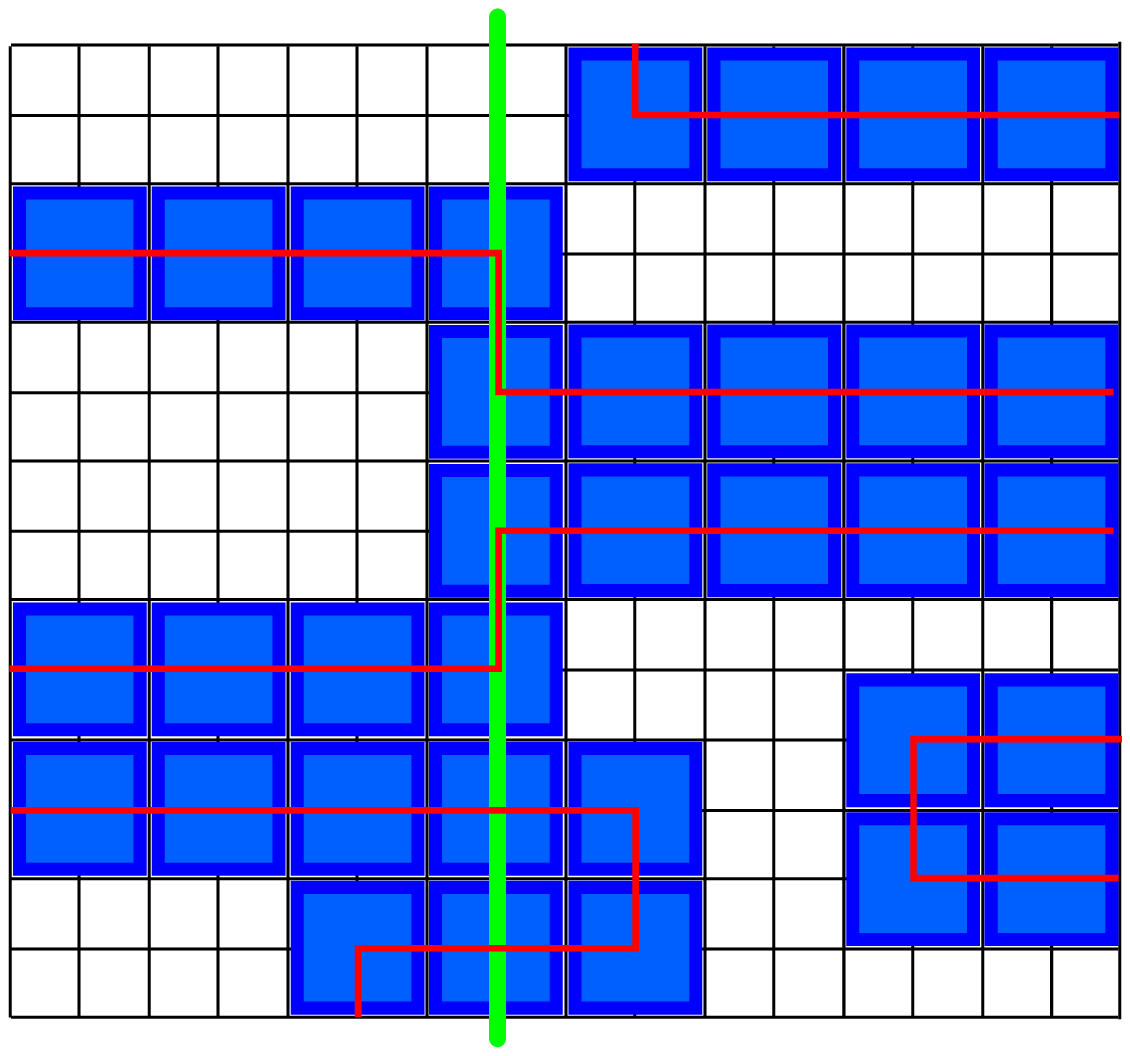}
 \caption{\label{mperfect}\small The green axis-parallel cut line $\ell$
intersects the tour $T$ (shown in red) in 4 connected components and
intersects the scan ranges (shown in blue), centered at scan points,
in a single chain of $k \times k$ disks.  (Here, $k=2$.)}
\end{figure}

The structure theorem is proved by showing the following lemma:

\begin{lemma}
For any fixed $m=\lceil 1/\epsilon \rceil$ and any choice of
$(T,{\cal S}(T))$, one can add a set of {\em doubled bridge
segments}, of total length $O(|T|/m)$, to $T$ and a set of $O(|{\cal
S}(T)|/m)$ {\em bridging scans} to ${\cal S}(T)$ such that the
resulting set, $(T_G,{\cal S}(T_G))$, is $m$-guillotine, with points
${\cal S}(T_G)$ on tour $T_G$ and with $T_G$ containing an Eulerian
tour of ${\cal S}(T_G)$.
\end{lemma}

\begin{proof}
For a vertical line, $\ell_x$, through coordinate $x$, let $f_1(x)$
denote the length of the {\em $m$-span} of $\ell_x$ with respect to
$R$ and $T$: $f_1(x)=0$ if, within $R$, $\ell_x$ intersects $T$ in
at most $2m$ connected components; otherwise, if $\ell_x$ intersects
$T\cap R$ in components $c_1,c_2,\ldots,c_K$, $K>2m$, then $f_1(x)$
is the distance (along $\ell_x$) from component $c_m$ to component
$c_{K-m+1}$.  Similarly, we define the length, $f_2(x)$, of the {\em
$m$-scan-span} of $\ell_x$ with respect to $R$ and ${\cal S}(T)$:
$f_2(x)=0$ if, within $R$, $\ell_x$ intersects at most $2m$ of the
scan disks ($k\times k$ squares) centered at points ${\cal S}(T)$;
otherwise, if $\ell_x$ intersects scan-disks $D_1,D_2,\ldots,D_K$,
$K>2m$, then $f_2(x)$ is the distance (along $\ell_x$) from disk
$D_m$ to disk $D_{K-m+1}$.  We think of $f_1(x)$ as the ``cost'' to
construct a vertical bridge for $T$ at position $x$, and $f_2(x)$ as
the ``cost'' to add a sequence (chain) of scan disks, centered along
$\ell_x$, and a detour of the tour (of length $O(k)$ per disk) that
visits their centers.  (We similarly define costs $g_1(y)$ and
$g_2(y)$ for constructing bridges along a horizontal cut $\ell_y$
through coordinate $y$.)  The total cost associated with a vertical
cut at $\ell_x$ is proportional, then, to $f_1(x)+f_2(x)$, and the
cost of a horizontal cut at $\ell_y$ is proportional to
$g_1(y)+g_2(y)$.  (Here is where we are using the fact that the
total cost is a linear combination of tour length and number of
scans, with a fixed bound, $c$, on the relative cost ratio of length
versus number of scans and that the scan disk has constant
size~$k$.)

In order to charge off the cost of constructing bridges along the
$m$-span and adding a chain of scan disks (and a subtour linking
them) along the $m$-scan-span, we use the notion of ``chargeable
length'' based on the ``$m$-dark'' length of $\ell_x\cap R$, with a
notion of ``$m$-darkness'' determined not just from the tour $T$,
but also from the set of scans ${\cal S}(T)$.  Specifically, a
subset $ab$ of $\ell_x\cap R$ is said to be $m$-dark with respect to
$T$ if for any $p\in ab$, the rightwards and leftwards rays from $p$
each cross at least $m$ (vertical) segments of $T$ before exiting
$R$. Similarly, $ab$ is said to be $m$-dark with respect to the scan
disks ${\cal S}(T)$ if for any $p\in ab$, the rightwards and
leftwards rays from $p$ each intersect at least $m$ scan disks
centered on ${\cal S}(T)$ before exiting $R$.  If a cut is made
along $\ell_x$, then the $m$-dark with respect to $T$ portion of the
cut can be charged off to the left/right sides of segments of $T$
lying to the right/left of $\ell_x$, distributing the charge to be
($1/m$)th to each of the $m$ segments first hit.  Similarly, the
portion of the cut that is $m$-dark with respect to ${\cal S}(T)$
can be charged off to the scan-disks (or, more precisely, to their
total perimeter, which is proportional (via constant $k$) to their
cardinality).

The key observation, then, is that there must exist a ``favorable''
vertical cut $\ell_x$ or horizontal cut $\ell_y$ for $R$ such that
the chargeable length of the cut is at least as long as the cost of
the cut.  This follows from the usual argument (\cite{m-gsaps-99}),
using the fact that $\int_{x\in R} (f_1(x)+f_2(x)) dx = \int_{y\in
R} h(y) dy$, where $h(y)$ is the chargeable length associated with
the horizontal cutt $\ell_y$, and assuming, without loss of
generality, that $\int_{x\in R} (f_1(x)+f_2(x)) dx \geq \int_{y\in
R} (g_1(y)+g_2(y)) dy$: There must exist a value $y^*$ where
$g_1(y^*)+g_2(y^*)\leq h(y^*)$, which then defines a favorable cut
$\ell_{y^*}$ for which the cost of constructing the $m$-span bridge
and the $m$-scan-span sequence of scan disks is chargeable to
lengths of $T$ and disks of ${\cal S}(T)$ in such a way that no
length of disk gets charged more than an amount proportional to
($1/m$)th of its length/count.

Once a favorable cut is found with respect to one rectangle $R$, the
cut partitions the problem into two subrectangles, and the argument
is recursively applied to each. The end result is an $m$-guillotine
subdivision of the original bounding box of $T$. \qed
\end{proof}

\begin{figure}[h]
\centering
    \comicII{.7\textwidth}{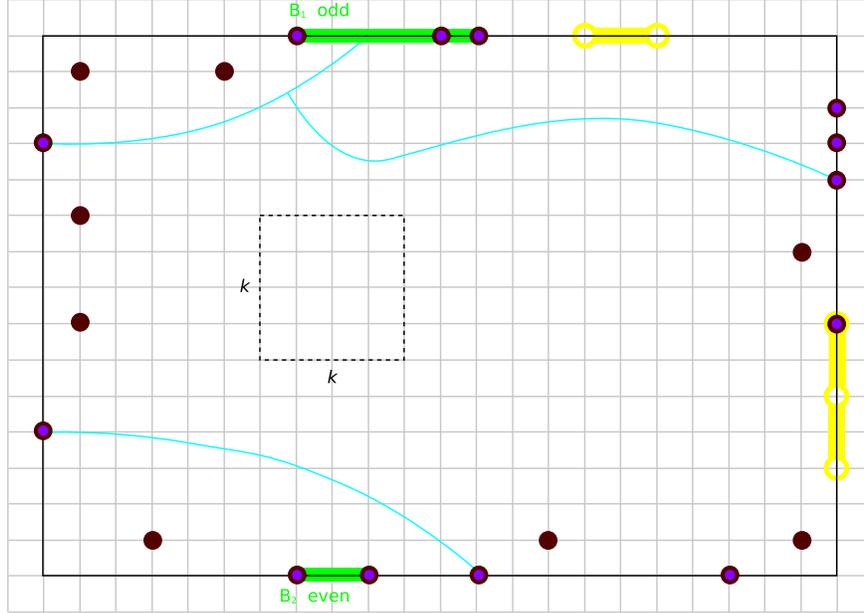}
 \caption{\label{boundinfo}\small An example for the boundary information within a rectangle R.
 The ``portals'' are depicted in violet, the bridges in green, the
 disk bridges in yellow, the scan positions such that a $k\times k$
square centered at each position intersects the corresponding side
of $R$ in brown and the connection pattern in turquoise.
 (The $k \times k$ scan range is shown by dotted lines.)}
\end{figure}

The algorithm is based on dynamic programming to compute an optimal
$m$-guillotine network.  A subproblem is specified by a rectangle,
$R$, with integer coordinates.  The subproblem includes
specification of {\em boundary information}, for each of the four
sides of $R$; see Figure~\ref{boundinfo}.  The boundary information
includes: (i) $O(m)$ integral points (``portals'') where the tour is
to cross/touch the boundary, (ii) at most one (doubled) bridge and
one disk-bridge (chain) per side of $R$, with each bridge having a
parity (even or odd) specifying the parity of the number of
connections to the bridge from within $R$, (iii) $O(m)$ scan
positions (from ${\cal S}(T)$) such that a $k\times k$ square
centered at each position intersects the corresponding side of $R$,
(iv) a connection pattern, specifying which subsets of the
portals/bridges are required to be connected within $R$.  There are
a polynomial number of subproblems.  For a given subproblem, the
dynamic program optimizes over all (polynomial number of) possible
cuts (horizontal or vertical), and choices of bridge, disk-bridge,
parity assignments, and compatible connection patterns for each side
of the cut.  The result is an optimal $m$-guillotine network, with
doubled bridges (so that it contains an Eulerian subgraph spanning
the nodes), and a scan-point set, ${\cal S}$, visited by the
network, such that the union of $k\times k$ squares centered at the
points ${\cal S}$ covers the input polygon $P$.  Since we know, from
the structure theorem, that an optimal tour $T$, together with a set
${\cal S}(T)$ of covering scan points can be converted into a tour
and scan-point set, $(T_G,{\cal S}(T_G))$, that are $m$-guillotine,
and we have computed an optimal such structure, we know that the
tour we extract from our computed network approximates optimal.  We
summarize:

\begin{theorem}\label{PTAS}
There is a PTAS for MWPDV lawn mowing of a (not necessarily
connected) set of pixels by a $k\times k$ square.
\end{theorem}

\paragraph{The Milling Variant.}
Our method applies also to the ``milling'' variant of the MWPDV, in
which the scans all must stay within the region $P$, provided that
$P$ is {\em simple} (no holes), and all of $P$ is reachable by a
scanner ($k\times k$ square).  We describe the changes that are
required. Subproblems are defined, as before, by axis-aligned
rectangles $R$. The difficulty now is that the restriction of $R$ to
$P$ means that there may be many ($\Omega(n)$) vertical/horizontal
chords of $P$ along one side of $R$. We can ignore the boundary of
$P$ and construct an $m$-bridge (which we can ``afford'' to
construct and charge off, by the same arguments as above) for $T$,
but only the portions of such a bridge that lie inside $P$ (and form
chords of $P$) are traversable by our watchman (since the other
portions are outside $P$). For each such chord, the subproblem must
``know'' if the chord is crossed by some edge of the tour, so that
connections made inside $R$ to a chord are not just made to a
``dangling'' component. We cannot afford to specify one bit per
chord, as this would be $2^{\Omega(n)}$ information. However, in the
case of a simple polygon $P$, no extra information must be specified
for the subproblem -- a chord is crossed by $T$ if and only if the
mower (scan) fits entirely inside the simple subpolygon on each side
of the chord; if the subpolygon outside of $R$, on the other side of
a chord, does not contain a $k\times k$ square, then we know that
the entire subpolygon is covered using scanned centered within $R$.
Exploiting this fact, the dynamic programming algorithm of our PTAS
is readily modified to the MWPDV milling problem within a simple
rectilinear polygon.

\begin{theorem}\label{PTAS-milling}
There is a PTAS for MWPDV milling of a simple rectilinear polygon by
a $k\times k$ square.
\end{theorem}

\section{Conclusion}\label{conclusion}

A number of open problems remain. Is it possible to remove the
dependence on the ratio $(r/a)$ of the approximation factor in our
algorithm for general MWPDV milling? This would require a
breakthrough for approximating minimum guard cover; a first step may
be to achieve an approximation factor that depends on $\log (r/a)$
instead of $(r/a)$.

For combined cost, we gave a PTAS for a lawn mowing variant, based
on guillotine subdivisions. The PTAS extends to the milling case for
simple rectilinear polygons.  We expect that the PTAS extends to
other cases too (circular scan disks, Euclidean tour lengths), but
the generalization to arbitrary domains with (many) holes seems
particularly challenging. Our method makes use of a fixed ratio
between scan cost and travel cost; as discussed in Figure~1, there
is no PTAS for the bicriteria version.

\subsection*{Acknowledgements}

We thank Justin Iwerks for several suggestions that improved the
presentation.


\small
\bibliographystyle{abbrv}
\bibliography{lit}

\begin{thebibliography}{10}

\bibitem{aab-mccps-06}
H.~Alt, E.~M. Arkin, H.~Br{\"o}nnimann, J.~Erickson, S.~P. Fekete, C.~Knauer,
  J.~Lenchner, J.~S.~B. Mitchell, and K.~Whittlesey.
\newblock Minimum-cost coverage of point sets by disks.
\newblock In {\em Proc. 22nd {ACM} Symposium on Computational Geometry}, pages
  449--458, 2006.

\bibitem{amp-lgvcp-10}
Y.~Amit, J.~S.~B. Mitchell, and E.~Packer.
\newblock Locating guards for visibility coverage of polygons.
\newblock {\em International Journal of Computational Geometry \&
  Applications}, to appear, 2010.

\bibitem{afm-aalmm-00}
E.~M. Arkin, S.~P. Fekete, and J.~S.~B. Mitchell.
\newblock Approximation algorithms for lawn mowing and milling.
\newblock {\em Computational Geometry: Theory and Applications},
  17(1-2):25--50, 2000.

\bibitem{bfks-esbgagp-10}
T.~Baumgartner, S.~P. Fekete, A.~Kr\"{o}ller, and C.~Schmidt.
\newblock Exact solutions and bounds for general art gallery problems.
\newblock In {\em Proc. {SIAM-ACM} Workshop on Algorithm Engineering and
  Experiments (ALENEX 2010)}, 2010.

\bibitem{bf-agdp}
C.~Baur and S.~P. Fekete.
\newblock Approximation of geometric dispersion problems.
\newblock {\em Algorithmica}, 30(3):451--470, 2001.

\bibitem{bgs-euped-01}
A.~Bhattacharya, S.~K. Ghosh, and S.~Sarkar.
\newblock Exploring an unknown polygonal environment with bounded visibility.
\newblock In {\em International Conference on Computational Science (1)},
  volume 2073 of {\em LNCS}, pages 640--648. Springer, 2001.

\bibitem{cn-owr-88}
W.-P. Chin and S.~Ntafos.
\newblock Optimum watchman routes.
\newblock {\em Proc. 2nd ACM Symposium on Computational Geometry},
  28(1):39--44, 1988.

\bibitem{cn-swrsp-91}
W.-P. Chin and S.~C. Ntafos.
\newblock Shortest watchman routes in simple polygons.
\newblock {\em Discrete {\&} Computational Geometry}, 6:9--31, 1991.

\bibitem{dm-aatnp-03}
A.~Dumitrescu and J.~S.~B. Mitchell.
\newblock Approximation algorithms for {TSP} with neighborhoods in the plane.
\newblock {\em J. Algorithms}, 48(1):135--159, 2003.

\bibitem{eh-ggt-06}
A.~Efrat and S.~Har-Peled.
\newblock Guarding galleries and terrains.
\newblock {\em Information Processing Letters}, 100(6):238 --245, 2006.

\bibitem{fms-mctc-09}
S.~P. Fekete, J.~S.~B. Mitchell, and C.~Schmidt.
\newblock Minimum covering with travel cost.
\newblock In {\em Proc. 20th International Symposium on Algorithms and
  Computation}, volume 5878 of {\em Lecture Notes in Computer Science}, pages
  393--402. Springer, 2009.

\bibitem{fs-pedv-08}
S.~P. Fekete and C.~Schmidt.
\newblock Polygon exploration with discrete vision.
\newblock {\em CoRR}, abs/0807.2358, 2008.

\bibitem{fs-lctnwdv-09}
S.~P. Fekete and C.~Schmidt.
\newblock Low-cost tours for nearsighted watchmen with discrete vision.
\newblock In {\em 25th European Workshop on Computational Geometry}, pages
  171--174, 2009.

\bibitem{fs-pedv-09}
S.~P. Fekete and C.~Schmidt.
\newblock Polygon exploration with time-discrete vision.
\newblock {\em Computational Geometry: Theory and Applications}, 43(2):148 --
  168, 2010.

\bibitem{hm-ascppipvlsi-85}
D.~S. Hochbaum and W.~Maass.
\newblock Approximation schemes for covering and packing problems in image
  processing and vlsi.
\newblock {\em Journal of the ACM}, 32(1):130--136, 1985.

\bibitem{ips-hpgg-82}
A.~Itai, C.~H. Papadimitriou, and J.~L. Szwarcfiter.
\newblock Hamilton paths in grid graphs.
\newblock {\em SIAM Journal on Computing}, 11(4):676--686, 1982.

\bibitem{tnccs-k-39}
R.~Kershner.
\newblock The number of circles covering a set.
\newblock {\em American Journal of Mathematics}, 61:665--67, 1939.

\bibitem{m-gsaps-99}
J.~S.~B. Mitchell.
\newblock Guillotine subdivisions approximate polygonal subdivisions: {A}
  simple polynomial-time approximation scheme for geometric {TSP}, {$k$-MST},
  and related problems.
\newblock {\em SIAM Journal on Computing}, 28:1298--1309, 1999.

\bibitem{m-ptast-07}
J.~S.~B. Mitchell.
\newblock A {PTAS} for {TSP} with neighborhoods among fat regions in the plane.
\newblock In {\em Proc. 18th Annual ACM-SIAM Symposium on Discrete Algorithms},
  pages 11--18, 2007.

\bibitem{o-agta-87}
J.~O'Rourke.
\newblock {\em Art Gallery Theorems and Algorithms}.
\newblock International Series of Monographs on Computer Science. Oxford
  University Press, New York, NY, 1987.

\bibitem{rt-rplbo-86}
P.~Rosenstiehl and R.~E. Tarjan.
\newblock Rectilinear planar layouts and bipolar orientations of planar graphs.
\newblock {\em Discrete {\&} Computational Geometry}, 1:343--353, 1986.

\bibitem{dkdk-t-49}
L.~F. T{\'o}th.
\newblock {\"U}ber dichteste {K}reislagerung und d{\"u}nnste
  {K}reis{\"u}berdeckung.
\newblock {\em Commentarii Mathematici Helvetici}, 23(1):342--349, 1949.

\bibitem{blw-mvpdr-00}
I.~A. Wagner, M.~Lindenbaum, and A.~M. Bruckstein.
\newblock {MAC} vs. {PC}: Determinism and randomness as complementary
  approaches to robotic exploration of continuous unknown domains.
\newblock {\em {ROBRES}: The International Journal of Robotics Research},
  19(1):12--31, 2000.

\end{thebibliography}

\end{document}